\documentclass{vldb}
\usepackage{graphicx}
\usepackage{balance}  
\usepackage{verbatim}
\usepackage{comment}
\usepackage{enumerate}
\usepackage{color}
\usepackage{xcolor}
\usepackage{relsize}
\usepackage{graphicx}
\usepackage{tabularx}
\usepackage{booktabs}
\usepackage{listings}
\usepackage{amssymb}
\usepackage{wrapfig}
\usepackage{algorithm}
\usepackage[utf8]{inputenc}
\usepackage{mathtools}
\usepackage{lipsum} 
\usepackage{CJKutf8}
\usepackage{listings}
\usepackage[hidelinks]{hyperref}

\hypersetup{
    colorlinks   = true,
    citecolor    = cyan,
    linkcolor    = magenta,
    urlcolor    = blue
}
\definecolor{cobalt}{rgb}{0.0, 0.28, 0.67}

\newtheorem{proposition}{Proposition}

\definecolor{codegreen}{rgb}{0,0.6,0}
\definecolor{codegray}{rgb}{0.5,0.5,0.5}
\definecolor{codepurple}{rgb}{0.58,0,0.82}
\definecolor{backcolour}{rgb}{0.95,0.95,0.92}

\lstdefinestyle{mystyle}{
    frame = single,
  commentstyle=\color{codegreen},
  keywordstyle=\color{magenta},
  numberstyle=\tiny\color{codegray},
  stringstyle=\color{codepurple},
  basicstyle=\footnotesize,
  breakatwhitespace=false,         
  breaklines=true,                 
  captionpos=b,                    
  keepspaces=true,                 
  numbers=left,                    
  numbersep=5pt,                  
  showspaces=false,                
  showstringspaces=false,
  showtabs=false,                  
  tabsize=2
}

\vldbTitle{Distributed Cross-Blockchain Transactions}
\vldbAuthors{Dongfang Zhao and Tonglin Li}
\vldbDOI{https://doi.org/10.14778/xxxxxxx.xxxxxxx}
\vldbVolume{12}
\vldbNumber{xxx}
\vldbYear{2020}

\begin{document}

\title{Distributed Cross-Blockchain Transactions}

\numberofauthors{2} 

\author{
\alignauthor
Dongfang Zhao\\
       \affaddr{University of Nevada, Reno}\\
       \email{dzhao@unr.edu}
\alignauthor
Tonglin Li\\
\affaddr{Lawrence Berkeley National Laboratory}\\
\email{tonglinli@lbl.gov}
}

\maketitle

\begin{abstract}

The interoperability across multiple or many blockchains would play a critical role in the forthcoming blockchain-based data management paradigm.
In particular, how to ensure the ACID properties of those transactions across an arbitrary number of blockchains remains an open problem in both academic and industry:
Existing solutions either work for only two blockchains or requires a centralized component,
neither of which would meet the scalability requirement in practice.
This short paper shares our vision and some early results toward scalable cross-blockchain transactions.
Specifically, we design two distributed commit protocols and, both analytically and experimentally, demonstrate their effectiveness.
\end{abstract}

\section{Introduction}

\subsection{Motivation}

A blockchain offers an immutable, decentralized, and anonymous mechanism for transactions between two entities on the same blockchain. 
Blockchain was not originally designed for online transactional processing (OLTP) workloads;
instead, it aimed to offer an autonomous and tamper-proof ledger service among mutually-distrusted parties and therefore, 
early blockchain systems can deliver only mediocre transaction throughput.

One natural question is whether and how we can adopt blockchains to efficiently handle OLTP workloads such that both autonomy and performance can be achieved at the same time. 
Indeed, much recent work focuses on this direction:
in~\cite{mhindi_vldb19,jwang_nsdi19}, authors advocate to leverage blockchains for OLTP workloads with various optimizations (e.g., sharding \cite{hdang:sigmod19}, sidechains \cite{sidechain}) to boost up the transaction throughput of blockchains,
such that blockchains would deliver similarly high performance as relational database systems (RDBMS) and turn to be a competitive alternative to the latter as a general-purpose data management system.

There is yet another critical, often overlooked, issue that must be addressed before blockchains can be widely adopted as a general data management paradigm:
the interoperability across heterogeneous blockchains.
While SQL along with the underlying distributed transaction handling are available between different vendors' RDBMS implementations, 
no such interface or general mechanism exists for blockchains.
Recent attempts (e.g., Cosmos~\cite{cosmos}) on such cross-blockchain transactions are all \textit{ad hoc} and exhibits poor scalability due to the centralized (physical or virtual) broker.

\subsection{Challenges}

We list four outstanding limitations exhibited by state-of-the-art cross-blockchain solutions:

\textbf{(1) Centralized Broker.} The transactions between heterogeneous blockchains are managed by a third-party, usually implemented as another blockchain (it is called a \textit{hub} in Cosmos).
This is against the decentralization principle of blockchains:
the broker would become a performance bottleneck, a single-point-of-failure, a target of security attacks.
Similarly, a recent work called AC3 \cite{vzakhary:arxiv19} employs an extra component (known as \textit{witness blockchain}) as a central authority to govern the cross-chain operations.
Although the witness blockchain is comprised of the nodes from existing blockchains, 
still, these \textit{virtual} nodes on the witness blockchain become the critical components of the entire ecosystem and, again, break the very core principle of blockchains.
    
\textbf{(2) Two-Party Transactions.} The protocols used by existing cross-blockchain systems stem from the \textit{sidechain protocol} \cite{sidechain},
    which was originally designed for transferring assets between Bitcoin \cite{bitcoin} and another cryptocurrency.
    The sidechain protocol speaks of nothing about three- or multi-party transactions;
    in fact, Cosmos only supports transferring assets between Bitcoin \cite{bitcoin} and Ethereum \cite{ethereum}.
A more recent line of works \cite{maurice:vldb19,mherlihy:podc18} are based on two-party Atomic Cross-Chain Swaps (ACS);
however, ACS cannot guarantee the atomicity of the multi-blockchain transaction as a whole.

\textbf{(3) Performance.} 
The sidechain protocol \cite{sidechain} took hours, if not days, to commit a single cross-blockchain transaction. 
    The main reason for this is due to the possible branches from the participating blockchains.
    In any participating blockchain,
    only one (i.e., the longest one) branch will remain valid, and any transactions from the shorter branches will rollback.
    This is not a problem if all of the transaction parties are from the same blockchain;
    But for cross-blockchain transactions,
    deliberate actions need to be taken.
    
\textbf{(4) Conventional Distributed Transactions.} 
One could argue that why not applying existing approaches of distributed transactions to multi-party blockchains?
The short answer is that the conventional wisdom did not assume the participant to proactively ``rollback'' its own decision,
which is not uncommon in blockchains.
For instance, in the conventional 2PC protocol~\cite{2pc}, when a participant replies a \textsc{ready-to-commit} message to the coordinator,
we assume that the decision is final and we can proceed to the next phase of the protocol.
In blockchains, however, the \textsc{ready-to-commit} message can be revoked by the participant later on, even \textbf{after} the transaction is completed only because the transaction happens to reside on a branch that is suppressed by a longer branch.
There was little study on such ``regrettable'' behavior of blockchains in the literature of distributed transactions.

\subsection{Contributions}

For completeness, Table~\ref{tbl:popular_protocols} summarizes candidate solutions with respect to two important properties regarding cross-blockchain transactions.
As we can see, existing works are limited to centralized design (i.e., the requirement of a hub), or the potential blocking, or both.
In our prior work~\cite{dzhao:cidr20}, we presented the roadmap toward cross-blockchain transactions, named CBT, to overcome the above limitations.

\begin{table}[th]
    \caption{Popular Cross-Blockchain Transaction Protocols.}
    \label{tbl:popular_protocols}
    \centering
    \begin{tabular}{ lll }
        \toprule
         &  Blocking & Nonblocking     \\    
        \midrule
        Centralized         & Sidechain~\cite{sidechain}        & AC3~\cite{vzakhary:arxiv19}     \\  \hline
        Distributed     & 2PC~\cite{2pc} & CBT~\cite{dzhao:cidr20} \\
        \bottomrule
    \end{tabular}
\end{table}

This paper is the first step toward the goals proposed in~\cite{dzhao:cidr20}.
Specifically, we will present a set of nonblocking distributed commit protocols designed for multi-party cross-blockchain transactions (\S\ref{sec:protocol}).
We will also present some preliminary results of these protocols (\S\ref{sec:result}),
followed by some discussions on our future work (\S\ref{sec:conclusion}).

\section{Protocols}
\label{sec:protocol}

\subsection{System Models and Assumptions}
\label{sec:protocol_model}

We assume the nodes follow a \textit{crash failure} model.
That is, there are no \textit{arbitrary} failures from the underlying blockchains and their participants.
We make this strong assumption as a starting point for this direction of research;
a \textit{Byzantine failure model} will be discussed in the future work.
Furthermore, we assume the crashed node will eventually be recovered and can be replaced by a functional node in a reasonable time, denoted by $f$.
Moreover, during a single transaction, the failures will not happen indefinitely but for finite times denoted by $\lambda$.

We assume the network transfer can be delayed but not indefinitely:
the communication is asynchronous and persistent.
That is, the messages can be \textit{eventually} delivered in a reasonable time.
The latency of the network is denoted by $\tau$.

We assume that a blockchain can \textit{finalize} the main branch in \textit{finite} time, after which the transactions cannot be rolled back.
In Bitcoin, for example, the \textit{pending} time is about one hour---six blocks of transactions.
We denote the average pending time for $C_i$ is $\delta_i$,
which also includes the \textit{waiting} time for a transaction to be picked up by the system.

We assume there is an effective programmable way for different blockchains to communicate.
This is mostly true for new blockchain implementations with \textit{smart contracts}.
For those old systems, e.g., Bitcoin, that do not support smart contract,
we assume a proxy is available on such systems for the cross-blockchain communications.

\textbf{Notations.} 
We denote the set of blockchains as $\mathcal{C} = \{C_i\}$, where each $C_i$, $i \in \mathbb{Z}^+$, represents a specific blockchain in the consortium of blockchains.
We use $\mathcal{C}_{-i}$ to denote the complement set $\mathcal{C} \setminus \{C_i\}$, following the naming convention in game theory.
The cardinality, or order, of the set,
i.e., $|\mathcal{C}|$, indicates the total number of blockchains involved in the transaction.
Each blockchain $C_i$ comprises a series of linked blocks, denoted as $B_i^j$,
where the superscript $j$ indicates the \textit{index} of the block on blockchain $C_i$. 
Each block is filled with a series of transactions,
denoted by $T_k$, where $k$ implies a universally unique identifier (UUID) of each transaction since the inception of the blockchain consortium.
It should be clear that, however,
although $k$ is unique globally,
it will appear at least once on each $C_i$ and possibly more than once if $C_i$ has branches during the processing of $T_k$.
For each $C_i$, there is a corresponding set $N_i \subseteq N$ denoting the set of \textit{nodes} having joined the network of blockchain $C_i$.
It is possible that a node joining multiple blockchains:
$n \in N_i$ and $n \in N_j$, $i \not= j$.

\textbf{Metrics.}
\textit{Throughput} is, arguably, the most popular metric in evaluating the performance of blockchains.
As in many other areas, the throughput of transactions is defined as the number of transactions completed in a time unit, usually in a second.
What is less used, or somewhat overlooked, metric, is the \textit{latency} of a transaction,
measuring the lifespan of a \textit{single} transaction in the blockchain systems.
We argue that latency is a more interesting metric from a user's standpoint:
she cares more about when her transaction is completed than how many concurrent transactions can be handled by the system \textit{per se}, concerned with by the system admin.

\subsection{Synchronous Cross-Blockchain Transactions Protocol (SBP)}

The first protocol is called Synchronous cross-Blockchain transactions Protocol (SBP) that is designed to strictly enforce the ACID properties of cross-blockchain transactions.
The targeting workloads include those that need to follow strong consistency models such as financial transactions.
As a trade-off, the performance, especially the latency, is not at the high end of the spectrum of candidate protocols.

SBP respects each individual blockchain's own branches and delays the global commit until no single blockchain can unilaterally rollback the transaction.
As the conventional wisdom in distributed commit protocols,
a specific blockchain initiates the multi-party transaction.
In the literature, this initiator is usually called a \textit{coordinator},
although we want to point out that this coordinator can be any participant $C_i$ in the pool $\mathcal{C}$.
Many \textit{leader election} algorithms can be applied to select the coordinator with the \textit{proxies} on $C$'s.
The specific node $n \in N_i$ serving as the \textit{endpoint} for the inter-blockchain communication can also be arbitrarily selected as long as the following conditions are met:
(i) other nodes $N_i \setminus \{n\}$ are aware of the role of $n$ and
(ii) all the intra-blockchain transaction updates have been applied to $n$.

Suppose $C_i$ initiates a transaction $T_k$ among all elements in $\mathcal{C}$, and $|\mathcal{C}| \geq 3$.
We will start describing the protocol in the \textbf{civil case}.
\begin{itemize}
    \item[\textbf{Phase I}] First, $C_i$ broadcasts a \textsc{precommit} message to (the proxies of) $\mathcal{C}$.
It should be clear that $C_i$ in this case serves as both the \textit{coordinator} and a (local) \textit{participant}.
$C_i$ then waits for a \textsc{Ready} reply from each blockchain in $\mathcal{C}$.
A blockchain $C_j \in \mathcal{C}$ (again, $j = i$ is allowed, implying a local message) replies a \textsc{ready} message to $C_i$ after (i) all prerequisites are satisfied, e.g., the balance is higher than the funds to be deducted in a cryptocurrency application, and 
(ii) more importantly, the entity is \textit{locked}.
The second action is crucial to avoid double-spending issues.

\item[\textbf{Phase II}] 
Second, $C_i$ braodcasts a \textsc{commit} message to $\mathcal{C}$ and waits for a \textsc{done} reply from each element in $\mathcal{C}$.
A blockchain $C_j \in \mathcal{C}$ carries out its local operation, and wait for $\delta_j$ before returning a \textsc{done} message to $C_i$.
The participants then should unlock the entities.
Once $C_i$ receives $|\mathcal{C}|$ \textit{done} replies, $T_k$ is marked completed.

\end{itemize}

Therefore, the civil case of SBP runs much like a 2PC protocol except for the introduction of pending time $\delta_j$.
The period enforced by $\delta_j$ can only preclude the possible branches in blockchains,
and yet cannot avoid the possible blocking in the uncivil case where nodes do fail (up to crash failures) incur possible blocking.
One way to fix that is to introduce an additional phase, essentially extending the protocol into three phases,
which has been extensively studied in the literature and is not a practical approach due to unacceptable performance.
What we propose to overcome the blocking issue is more lightweight: taking a passive heartbeat approach to effectively detect node failures.
It should be noted that this approach becomes effective only because in cross-blockchain transactions each node is essentially a set of nodes, i.e., $N_i$ for $C_i$,
such that if the original proxy node $n \in N_i$ fails, we can quickly re-select $n' \in N_i$ to continue the SBP protocol.

Formally, suppose $n \in N_i$ is the endpoint of $C_i$, the proxies on other nodes $N_i \setminus \{n\}$ run a \textit{heartbeat} probe to $n$, whose interval is denoted as $\sigma_i$.
Let $\overline{\sigma} = \mathtt{sup}\{\sigma_i, 1 \leq i \leq |\mathcal{C}|\}$,
it is not hard to see that SBP can be blocked by up to $\overline{\sigma}$.
In practice, we can set $\overline{\sigma} \ll \underline{\delta}$,
where $\underline{\delta} = \mathtt{inf}\{\delta_j, 1 \leq j \leq |\mathcal{C}|\}$,
such that the heartbeat overhead is negligible.

\subsubsection{Correctness}

\textbf{Atomicity.}
SBP takes a conservative approach to commit the requested transaction. 
At any point during the two-phase protocol, 
any states other than the expected ones mentioned in the protocol narrative results in a global \textsc{abort}.
A more subtle yet rare case is that no qualified node can be found after the \textit{heartbeat} protocol detects a crash failure,
in which case the entire SBP also aborts the transaction.

\textbf{Consistency.}
The changes incurred by the transaction would be invisible to users until the $C_i$ marks the completion of the transaction.
Thus, SBP implements a strong consistency model, 
there are no dirty-write or repeated-read issues during the course of distributed transaction processing.
Indeed, this strong consistency is attributed to the locking approach with the price of suboptimal performance in transaction latency.
We will speak more about performance in the complexity discussion shortly.

\textbf{Isolation.}
This can be trivially verified by the fact that locking and unlocking are implemented correctly, as discussed in the protocol.

\textbf{Durability.}
Updates are persisted on all the nodes in each involved blockchain.

\subsubsection{Analysis}

We will show that the number of messages is asymptotically polynomial to the number of nodes among all blockchains.

\begin{proposition}[Number of messages]
The total number of messages passed is $\mathcal{O}(\lambda |N|)$.
\end{proposition}
\begin{proof}
Obviously, the maximal number of messages are sent when the nodes are failed repeatedly for \textit{finite} times, and the transaction eventually completes.
It is crucial to note that the failure can happen for limited times because otherwise, our assumption would not hold (cf.~\S\ref{sec:protocol_model}).

In phase I, the total number of messages between elements in $\mathcal{C}$ is 
\begin{equation*}
\begin{aligned}
    & \overbrace{2 \cdot \lambda \cdot (|\mathcal{C}| - 1)}^{\texttt{inter-blockchain}} + \overbrace{\sum_{C_i \in \mathcal{C}} (|N_i| - 1)}^{\texttt{intra-blockchain}} \\
=   & \quad 2 \cdot \lambda \cdot (|\mathcal{C}| - 1) + |N| - |\mathcal{C}|\\
=   & \quad |N| + (2\lambda - 1)|\mathcal{C}| - 2\lambda\\
\leq & \quad 2\lambda|N|. \quad (\texttt{since } |\mathcal{C}| \leq |N| \texttt{ and } \lambda \geq 0)
\end{aligned}
\end{equation*}
The messages in phase II can be similarly calculated.
The total number of messages is thus less than $4\lambda|N|$, proving the proposition.
\end{proof}

Thus, the number of messages is asymptotically polynomial to the number of nodes among all blockchains.
We then study the theoretical upper bound of the transaction latency.

\begin{proposition}[Latency upper bound]
The longest period for a single transaction, i.e., the latency, is bounded by $4\tau + \lambda (f + \overline{\delta})$, where $\overline{\delta} = \mathtt{sup}\{\delta_i, 1 \leq i \leq |\mathcal{C}|\}$.
\end{proposition}

\begin{proof}
The latency of phase I is calculated as
\[
\Delta_1 = 2 \cdot \tau + \lambda_1 \cdot f
,\]
and the latency of phase II is bounded by
\[
\Delta_2 = 2 \cdot \tau + \lambda_2 \cdot f + \underbrace{\sum_{n \in N_i}\delta_i}_{\lambda_2}
,\]
where $n$ indicates a failed node and $\lambda = \lambda_1 + \lambda_2$.
Therefore, the overall latency
\begin{equation*}
\begin{aligned}
\Delta  & = \Delta_1 + \Delta_2 \\
        & \leq 4\tau + (\lambda_1 + \lambda_2)f + \lambda_2 \overline{\delta}\\
        & \leq 4\tau + \lambda (f + \overline{\delta})
.\end{aligned}
\end{equation*}
\end{proof}

In practice, $\tau$ can be easily measured, in terms of milliseconds;
$f$ usually takes a few seconds, e.g., to reboot the failed node;
$\delta$ is also well understood: in Bitcoin, for instance, it takes hours to finalize a transaction.
However, it is not trivial to estimate $\lambda$ other than keep an empirical log over the failure rate.
We want to point out that a Poisson distribution can become a handy tool for quickly estimating the transaction delay.
That is, the probability of $k$ failures can be estimated by $\displaystyle \frac{\lambda^k e^{-\lambda}}{k!}$,
where $e$ is Euler's number.

\subsection{Redo-Log-Based Blockchain Protocol (RBP)}

While SBP discussed in the previous section achieves strong consistency,
the price is the somewhat long delay.
Therefore, SBP is ideal for those time-insensitive applications that are required to guarantee ACID properties.
This section studies the other end of the spectrum:
what if the workload is highly time-sensitive and can tolerate temporary inconsistencies.
That is, the applications, such as emails, can accept an \textit{eventual consistency} semantic.
To this end, we design a distributed commit protocol, namely RBP, 
following the spirit of redo-logs that has been extensively studied in databases.

RBP makes a key change to the way how participants reply the \textsc{done} messages back to the coordinator $C_i$.
Instead of waiting for a period of $\delta_j$,
$C_j$ replies $C_i$ right after the local updates are completed.
Indeed, the question then becomes what if $C_j$ decides to cut off the branch comprising the completed transaction $T_k$ between $C_i$ and $C_j$'s later on?
To this end, blockchain $C_i$ maintain a \textit{sliding window} that records recent transactions completed in the past $\delta_i$ period.
The rationale is that if any of these \textit{pending} transactions are on the path of a shorter branch of $C_i$,
$C_i$ can take according actions such as returning the transactions back to the request pool, 
or immediately rescheduling the transactions.
RBP takes the former approach: transactions on the shorter branches are recycled back into the pool of requests.
Note that we cannot construct \textit{complement} transactions to \textit{undo} the changes because those transactions are invisible to the main branch of $C_i$.

Evidently, RBP still meets the \textit{atomicity} requirement:
there is no ``partial'' transaction committed.
It is also trivial to check that both \textit{isolation} and \textit{durability} hold in RBP.
For \textit{consistency}, RBP implements an eventual consistency semantics:
the transactions on shorter branches will eventually be reprocessed.
We conclude this section with a more detailed quantitative study in the following.

We are particularly interested in the improved latency paid by the weak consistency semantics.
Let the transactions in the sliding window of $C_i$ be $\mathcal{T}_i$.
Consequently, the throughput of blockchain $C_i$ can be calculated by $\displaystyle \frac{|\mathcal{T}_i|}{\delta_i}$.
Because of the possible \textit{cascading effect} implied by the $C_i$'s indeterministic branching behavior,
we cannot derive an upper bound over the latency of a transaction $T_k \mathcal{T}_i$.
However, if no branching happens during $T_k$, the latency can be as low as $4\tau + \lambda f$.
Recall that both $\tau$ and $f$ are orders of magnitude smaller than $\delta$, and $\lambda$ represents a few failed nodes in a time unit;
therefore, RBP is expected to deliver a significantly smaller latency than SBP.
Again, this gain is traded by the (strong) consistency.

\section{Preliminary Results}
\label{sec:result}

We have implemented RBP protocol as well as two baseline protocols, i.e., 2PC \cite{2pc} and AC3 \cite{vzakhary:arxiv19}, on the BlockLite system \cite{xwang:cloud19}.
The source code is accessible at Github~\cite{github_cbt}.
The source code is written with Java of JDK 1.7.
The codebase comprises about 5,190 lines of code. 
The codebase has three major components: 
(i) the blockchain implementation including protocols and utilities; 
(ii) the network component including the communications among coordinator and participants;
and (iii) the graphic user interface developed with Java Swing.
A technical report on an earlier version of the system can be found at~\cite{xwang:arxiv20}.

The transaction data sets used for evaluation are ETC20 and TPC-H.
We feed up to three million transactions to the three protocols (RBP, 2PC, and AC3) in a 64-blockchain environment.
For 2PC, we set it up as the ``ideal case'' where no failures take place during the experiment;
it is the upper-bound performance one can best expect from 2PC.
The point is to show the overhead incurred by our proposed CBT compared to such upper-bound performance.
For AC3, we arbitrarily select one blockchain as the ``hub'', or ``witness blockchain'' as in the literature.
Because of AC3's centralized hub, we expect the performance and scalability will be affected at some point, e.g., a larger number of blockchains.
Both 2PC and CBT show (almost) linear scalability because no centralized component exists in the system.
Results show that RBP incurs insignificant overhead (compared with baseline 2PC) at small/medium scales:
3.6\% -- 4\% on 2--32 blockchains;
then the overhead is negligible on 64 blockchains.
Compared with 2PC and RBP, AC3 starts to fall behind on eight blockchains due to its ``hub'' design.

\section{Conclusion and Future Work}
\label{sec:conclusion}

As a concluding remark, we want to reemphasize that the future blockchain-based data management paradigm must be equipped with effective cross-blockchain transactions at arbitrary scales,
which cannot be realized without a scalable distributed commit protocol specifically designed for transactional workloads, namely cross-blockchain transaction (CBT).
The role of the CBT protocols for the future blockchain-based paradigms can be considered as the analogue to:
TCP/UDP for network systems, HTTP for web servers, FTP for file servers, and so forth.
What is presented in this short paper is only one of the first steps toward the future standardization when the time comes for many heterogeneous blockchain systems to jointly work as a coordinated service or platform.

In addition to industry workloads represented by ETC20 and TPC-H already tested with CBT,
we are working with a team at the Lawrence Berkeley National Laboratory on deploying CBT to one of the largest supercomputers Cori~\cite{cori} to:
(i) Justify the feasibility of blockchains for distributed caching in high-performance computing systems;
(ii) Evaluate the scalability and performance of CBT for huge scientific workloads (e.g., the data provenance of astronomy applications); and
(iii) Quantify the energy efficiency of large-scale blockchain deployment.
Some preliminary results on high-performance blockchains can be found at~\cite{aalmamun_sc19}.

\section*{Acknowledgement}

This work is in part supported by the U.S. Department of Energy under contract number DE-SC0020455.
This work is also supported by a Google Cloud award and an Amazon research award.
The authors are grateful for the valuable discussion with Mohammad Sadoghi (University of California, Davis) on an earlier version of this work.

\bibliographystyle{abbrv}
\bibliography{vldb}

\end{document}